\documentclass{article}
\usepackage{amsmath,amssymb,amsthm}
\usepackage{amsmath,amssymb}
\usepackage[OT2,T1]{fontenc}
\usepackage{pictexwd,dcpic}
\usepackage[pdftex]{graphicx}
\usepackage{paralist}
\newtheorem{theorem}{Theorem}
\newtheorem{lemma}[theorem]{Lemma}
\newtheorem{corollary}[theorem]{Corollary}

\newtheorem{definition}[theorem]{Definition}

\renewcommand{\epsilon}{\varepsilon}

\newcommand{\nfa}{\textnormal{NFA}}
\newcommand{\dfa}{\textnormal{DFA}}

\DeclareMathOperator{\crank}{r}
\DeclareMathOperator{\dpw}{dpw}
\DeclareMathOperator{\snum}{s}
\DeclareMathOperator{\starheight}{h}

\newcommand{\NP}{{\bf NP}}
\newcommand{\PSPACE}{{\bf PSPACE}}

\newcommand{\EXPSPACE}{{\bf EXPSPACE}}
\newcommand{\trans}[1]{\stackrel{\!#1}{\rightarrow}}

\begin{document}

\title{Digraph Complexity Measures and Applications in Formal Language Theory}


\author
{Hermann Gruber\thanks{This paper is a completely revised and 
expanded version of research presented at the 
$4$th Workshop on Mathematical and Engineering Methods in 
Computer Science (MEMICS 2008), see~\cite{Gruber08b}.}
\thanks{Part of the work was done while the author was at 
Institut f\"ur Informatik, Justus-Liebig-Universit\"at Giessen, Germany.}\\
knowledgepark AG, Leonrodstr. 68, D-80636 M\"unchen, Germany}
\date{}
\maketitle

 \begin{abstract}
  We investigate structural complexity measures on digraphs, in particular
  the cycle rank. This concept is 
  intimately related to a classical topic in formal language theory, 
  namely the star height of regular languages. We explore this connection,
  and obtain several new algorithmic insights regarding both 
  cycle rank and star height. Among other results, we 
  show that computing the cycle rank is $\NP$-complete, 
  even for sparse digraphs of maximum outdegree~$2$. 
  Notwithstanding, we 
  provide both a polynomial-time approximation algorithm 
  and an exponential-time exact algorithm for this problem. 
  The former algorithm yields an $O((\log n)^{3/2})$-approximation 
  in polynomial time, whereas the latter yields the optimum solution, and 
  runs in time and space $O^*(1.9129^n)$ on digraphs of maximum 
  outdegree at most two.

  Regarding the star height problem, we identify a subclass 
  of the regular languages for which we can precisely determine the 
  computational complexity of the star height problem. Namely, the 
  star height problem for bideterministic languages is $\NP$-complete, 
  and this holds already for binary alphabets. Then we translate
  the algorithmic results concerning cycle rank to the bideterministic
  star height problem, thus giving a polynomial-time approximation
  as well as a reasonably fast exact exponential algorithm 
  for bideterministic star height.  
 \end{abstract}
\noindent
{\small
} 
\section{Introduction}
In the theory of undirected graphs, structural complexity measures 
for graphs, such as treewidth and pathwidth, have gained an important role,
both from a structural and an algorithmic viewpoint, see 
e.g.~\cite{Diestel06,KT05}. 
However, networks arising in some domains are more adequately 
modeled as having directed edges. Therefore in recent years, attempts
have been made to lift such measures and parts of the theory of 
undirected graphs to the case of 
digraphs. Several recent works show that, 
while there often exist partial analogues to the undirected case, 
the picture for digraphs is much more 
involved~\cite{Barat06,BDHKO11,GHKLOR09,KO11,Safari05}. 
We discuss some of these
measures, relate them to each other, and investigate their algorithmic 
aspects. Interestingly, we are able to show that all these 
complexity measures bound each other within a factor logarithmic 
in the order of the digraph, thus paralleling the case of 
undirected graphs~\cite{BGHK95}.  
We focus in particular on the 
cycle rank, a digraph complexity 
measure originally motivated by studies in 
formal languages~\cite{Eggan63}. Apparently, there is a 
renewed interest in this 
measure, as witnessed by recent research efforts~\cite{Alpert10,BCLMZ11,
GHKLOR09,Hunter11,LKM11}.

We obtain the following 
results on computing the cycle rank: The decision version of the 
problem is $\NP$-complete, and this remains true for graphs of 
maximum outdegree at most~$2$.
Previously, the problem was known to be $\NP$-complete on undirected
symmetric digraphs of unbounded degree, see~\cite{BDJKKMT98}. 
On the positive side, we design a polynomial-time 
$O((\log n)^{3/2})$-approximation algorithm, as well as 
an exact exponential algorithm algorithm computing the cycle 
rank of digraphs. If the given digraph is of bounded outdegree, the 
latter algorithm runs 
in time and space $O^*((2-\epsilon)^n)$, where~$n$ is the order of 
the digraph, and~$\epsilon$ is a constant depending on 
the maximum outdegree. For unbounded outdegree, 
the running time is still $O^*(2^n)$, whereas for 
maximum outdegree~$2$, we even attain a bound of $O^*(1.9129^n)$.
As a further application, we also obtain an exact algorithm
for the directed feedback vertex set problem on digraphs of
maximum outdegree~$2$, which runs within the same time bound.

Then we present applications of these findings to the theory 
of regular expressions.
The star height of a regular language is defined as the 
minimum nesting depth of stars needed in order to describe 
that language by a regular expression. 
Already in the $1960$s, Eggan~\cite{Eggan63} raised
the question whether the star height can be determined
algorithmically. It was not until~$25$ years
later that Hashiguchi found a rather complicated decidability
proof~\cite{Hashiguchi88}. Even today, 
the best known algorithm has doubly exponential running time, 
and is arguably still impractical~\cite{Kirsten10}.
Therefore, we study the complexity of the star height problem 
when restricted to a subclass of the regular languages.
We show that the star height problem for bideterministic languages
is $\NP$-complete, and this remains true when restricted to binary 
alphabets. Furthermore, we present both an efficient approximation
algorithm and an exact exponential algorithm for this problem. 
The key to these results are the corresponding 
algorithms for the cycle rank of digraphs 
mentioned above; also the above mentioned bounds carry 
over to this application in formal language theory.

The paper is organized as follows: After this introduction, we 
recall in Section~\ref{sec:prelim} some basic notions 
from graph theory and from automata theory. We study structural 
properties of the cycle rank of digraphs in Section~\ref{sec:cycle-rank}.
Section~\ref{sec:cr-algo} is devoted to algorithmic aspects 
of cycle rank.
Afterward, we apply these findings in Section~\ref{sec:star-height} 
to the star height problem on bideterministic languages.
We complete the paper in Section~\ref{sec:conclusion}
by showing up possible directions for further research.

\section{Preliminaries}\label{sec:prelim}

\subsection{Digraphs}
We assume familiarity with basic notions in graph theory, 
as contained in~\cite{Diestel06}, so we only fix the notation and 
a few specialties below.
A {\em digraph} $G=(V,E)$ consists of a finite set of {\em vertices}~$V$ 
and a set of {\em edges}~$E\subseteq V^2$. 

We refer to an edge of the form~$(v,v)$ as a 
{\em loop}; A digraph without loops is called {\em loop-free}.

The {\em outdegree} of a vertex $v$ is defined as the number
of vertices~$u$ such that~$(u,v) \in E$. The {\em total degree}
is defined as the number of distinct vertices~$u$ 
having~$(u,v) \in E$ or $(v,u)\in E$.  

If the edge relation of a digraph~$G$ is symmetric, we say~$G$ is an 
(undirected) {\em graph}. By taking the symmetric 
closure of the edge relation of a digraph, 
we obtain its undirected counterpart---of course, 
this is a many-to-one correspondence. 

For a subset of vertices $U\subseteq V$, 
let $G[U]$ denote the sub(di)graph {\em induced by} $U$, which is
obtained by restricting the vertex set of~$G$ to~$U$ and redefining
the edge set~$E$ appropriately. In this context, we will often 
use $G-U$ as a shorthand for $G[V\setminus U]$ and $G-v$ for 
$G[V\setminus\{v\}]$. A subset of vertices $U\subseteq V$ 
is {\em strongly connected} if for every $v\in V$ there is a 
(possibly empty) path from~$v$ to itself. Maximal strongly connected 
subsets of~$V$ are called {\em strongly connected components}; a strongly
connected subset~$S$ is {\em nontrivial} if the subdigraph 
$G[S]$ induced by~$S$ contains at least one edge (note that this also allows 
the case $S=\{v\}$ if~$v$ has a loop). A digraph is {\em acyclic}
if all of its strongly connected components are trivial.

\subsection{Formal Languages}
As with digraphs, we only recall 
some basic notions in formal
language and automata theory---for a thorough treatment, the 
reader might want to consult a textbook such as~\cite{HU79}. 
In particular, let~$\Sigma$ be a finite alphabet and~$\Sigma^*$ the set of all
words over the alphabet~$\Sigma$, including the empty word~$\lambda$.
The length of a word~$w$ is denoted by~$|w|$, where $|\lambda|=0$.
A {\em (formal) language\/} over the alphabet~$\Sigma$
is a subset of~$\Sigma^*$. 

The {\em regular expressions} over an alphabet~$\Sigma$ are defined
recursively in the usual way:\footnote{%
  For convenience, parentheses in regular expressions are sometimes
  omitted and the concatenation is simply written as
  juxtaposition. The priority of operators is specified in the usual
  fashion: concatenation is performed before union, and star before
  both product and union.}  $\emptyset$, $\lambda$, and every
letter~$a$ with $a\in \Sigma$ is a regular expression; and when~$r_1$
and~$r_2$ are regular expressions, then $(r_1+r_2)$, $(r_1\cdot r_2)$,
and $(r_1)^*$ are also regular expressions.  The language defined by a
regular expression~$r$, denoted by $L(r)$, is defined as follows:
$L(\emptyset)=\emptyset$, $L(\lambda)=\{\lambda\}$, $L(a)=\{a\}$,
$L(r_1+r_2)=L(r_1)\cup L(r_2)$, $L(r_1\cdot r_2)=L(r_1)\cdot L(r_2)$,
and $L(r_1^*)=L(r_1)^*$. For a regular expression $r$ over~$\Sigma$, 
the {\em star height}, denoted
by $\starheight(r)$, is a structural complexity measure inductively defined by:
$\starheight(\emptyset)=\starheight(\lambda)=\starheight(a)=0$, $\starheight(r_1\cdot r_2) = \starheight(r_1 + r_2) =
\max\left(\starheight(r_1),\starheight(r_2) \right)$, and $\starheight(r_1^*) = 1 + \starheight(r_1)$.  The
star height of a regular language $L$, denoted by $\starheight(L)$, is then
defined as the minimum star height among all regular expressions
describing $L$. 

It is well known 
that regular expressions are exactly as powerful as
finite automata, i.e., for every regular expression one can construct
an equivalent (deterministic) finite automaton and {\it vice
  versa}, see~\cite{HU79}. Finite automata are defined as follows: A {\em
  nondeterministic finite automaton\/} (\nfa) is a $5$-tuple
$A=(Q,\Sigma,\delta, q_0,F)$, where~$Q$ is a finite set of states,
$\Sigma$ is a finite set of input symbols,
$\delta:Q\times\Sigma\rightarrow 2^Q$ is the transition function,
$q_0\in Q$ is the initial state, and $F\subseteq Q$ is the set of
accepting states.  The {\em language accepted\/} by the finite
automaton~$A$ is defined as $L(A) =\{\,w\in \Sigma^*\mid
\delta(q_0,w)\cap F\neq\emptyset\,\}$, where~$\delta$ is naturally
extended to a function $Q\times\Sigma^*\rightarrow 2^Q$. A
nondeterministic finite automaton $A=(Q,\Sigma,\delta,Q_0,F)$ is {\em
  deterministic\/}, for short a \dfa, if $|\delta(q,a)| \le 1$, for
every $q\in Q$ and $a\in\Sigma$. In this case we simply write
$\delta(q,a)=p$ instead of $\delta(q,a)=\{p\}$.  Two (deterministic or
nondeterministic) finite automata are {\em equivalent\/} if they
accept the same language.

A deterministic finite automaton is 
{\em bideterministic}, if it has a single final state, and if the \nfa{} 
obtained by reversing all transitions and exchanging the roles of 
initial and final state is again deterministic---notice that, 
by construction, this \nfa{} in any case accepts the reversed language. 
A regular language~$L$ is {\em bideterministic\/}
if there exists a bideterministic finite automaton accepting $L$. 
These languages form a proper subclass of the regular languages~\cite{Angluin82}. 

\section{Cycle Rank of Digraphs}\label{sec:cycle-rank}

\subsection{Cycle Rank and Directed Elimination Forests}

Originally suggested in the $1960$s by Eggan and B\"uchi in the course of 
investigating the star height of regular languages~\cite{Eggan63}, 
the cycle rank is probably one of the oldest structural 
complexity measures on digraphs. In this section, 
we delve into the structural foundations of cycle rank.  
 
\begin{definition}\label{defn:cr}
The {\em cycle rank} of a directed graph 
$G=(V,E)$, denoted by $\crank(G)$, is inductively defined as
follows: If~$G$ is acyclic, then $\crank(G) = 0$.
If~$G$ is strongly connected and $E \neq \emptyset$, 
then $\crank(G) = 1 + \min_{v \in V}\{\,\crank(G-v)\,\}$.
If~$G$ is not strongly connected, then $\crank(G)$
equals the maximum cycle rank among all strongly connected components of~$G$.
\end{definition}
We note that the requirement $E \neq \emptyset$ in the above definition 
allows to differentiate between acyclic digraphs and (otherwise acyclic) 
digraphs with loops. 
We also remark that the cycle rank can be equivalently defined using 
decompositions, compare~\cite{McNaughton69}:
\begin{definition}\label{defn:elimination-forest} 
A {\em directed elimination tree} for 
a nontrivially strongly connected digraph~$G=(V,E)$ is a rooted tree 
$T=(\mathcal{T}, \mathcal{E})$ having the following properties:
\begin{enumerate}[a)]
\item
$\mathcal T \subseteq V \times 2^{V}$,
and if $(x,X)\in \mathcal T$, then $x \in X$.
\item
The root of the tree is $(v,V)$ for some $v \in V$.
\item
There is no pair distinct vertices of the form $(x,X)$ and $(y,X)$ in
  the forest.
\item If $(x,X)$ is a node in $T$, and $G[X]-x$ has $j\ge 0$ nontrivial 
strongly connected components $Y_1,\ldots,Y_j$, then $(x,X)$ has exactly
$j$ children of the form $(y_1,Y_1),\ldots (y_j,Y_j)$ for some
$y_1,\ldots,y_j\in V$.
\end{enumerate}
A {\em directed elimination forest} for a digraph $G$ with $k\ge 0$
nontrivial strongly connected components $C_1,\ldots C_k$, is a rooted forest 
consisting of directed elimination trees for $G[C_i]$, $1\le i\le k$.
\end{definition}
Figure~\ref{fig:digraph} illustrates this concept by an example.
\begin{figure}[t]
\begin{center}
  \centering
  \includegraphics{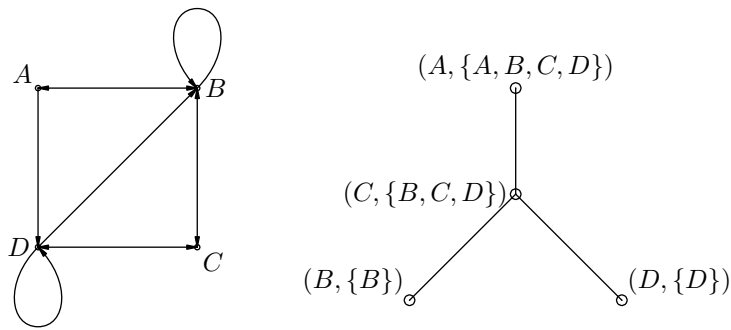}
  \caption{An example digraph and a directed elimination tree for it.}
  \label{fig:digraph}
\end{center}
\end{figure}
It is shown in~\cite{McNaughton69} that the minimum height among all directed 
elimination forests for~$G$ equals the cycle rank of~$G$. 
Interestingly, the concept of elimination forests was
rediscovered in the context of sparse matrix 
factorization, in~\cite{Schreiber82} 
for the undirected case and in~\cite{EL05} for the directed case. 

\subsection{Cycle Rank and other Digraph Complexity Measures}

We compare the cycle rank with two other structural complexity measures, namely 
weak separator number and directed pathwidth. The first 
measure is a generalization of separator 
number (see e.g.~\cite{BGHK95,RS86,Gruber10}) to digraphs: 
\begin{definition}
Let $G=(V,E)$ be a digraph and let $U \subseteq V$ be a 
set of vertices. A set of vertices~$S$ is a {\em weak balanced separator\/} 
for~$U$ if every strongly connected component of $G[U \setminus S]$ 
contains at most $\left\lceil\frac12|U-S|\right\rceil$ vertices.  
The {\em weak separator number\/} of~$G$, denoted by $\snum(G)$, 
is defined as the maximum size, taken over all subsets~$U\subseteq V$, 
among the minimum weak balanced separators for~$U$. 
\end{definition}
Some readers will feel that the above definition is a bit contrived 
because of the ceiling operator~$\lceil \cdot \rceil$. 
But this is an essential detail, as it guarantees that a 
digraph with a weak balanced separator of size~$k$ 
will always admit a weak balanced separator of size~$k+1$. 

In order to relate weak separator number and cycle rank, 
we need the following recurrence: 
For integers $k,n \ge 1$, let $R_k(n)$ be given by the recurrence
\[R_k(n) = k + R_k\left( \left\lceil \frac{n-k}{2} \right\rceil\right),\]
with $R_k(r_0) = r_0$ for $r_0\le k$.

\begin{lemma}\label{lem:cr-separator}
Let~$G$ be a loop-free digraph with~$n$ vertices and weak separator number 
at most~$k$. Then $\crank(G)\le R_k(n)-1$.
\end{lemma}
\begin{proof}
We generalize a proof given in~\cite{Gruber10} to the case of digraphs.

Let $G^\ell$ be the digraph obtained from~$G$ by adding 
self-loops to each vertex. Then $\crank\!\left(G^\ell\right) = \crank(G)+1$, 
so we may prove instead that $\crank\!\left(G^\ell\right)\le R_k(n)$. 

We prove the statement by induction on the order~$n$ 
of~$G^\ell$. The base cases $n \le k$ of the induction 
are easily seen to hold, since the cycle rank of a 
digraph is always bounded above by its order. 

For the induction step, assume $n > k$. 
As already mentioned, if $G^\ell$ admits a weak 
balanced separator of size at most~$k$, then 
it also has a weak balanced separator of size 
exactly~$k$. Let~$X$ be such a separator. 

Denote the strongly connected components of $G^\ell-X$ 
by $C_1, \ldots, C_p$. 
Then $\crank\!\left(G^\ell\right) \le k + \crank\!\left(G^\ell-X\right)$, and 
by definition of cycle rank, 
\[\crank\!\left(G^\ell-X\right) \le \max_{1\le i \le p}\crank\!\left(G^\ell[C_i]\right).\] 
As~$X$ is a weak balanced separator, we have 
$|C_i| \le \lceil\frac{n-k}{2}\rceil$ for 
$1\le i \le p$, so we can apply the induction 
hypothesis to obtain 
\[\max_{1\le i \le p}\crank\!\left(G^\ell[C_i]\right) \le R_k(\lceil\frac{n-k}{2}\rceil).\]
Putting these pieces together, we have 
$\crank\!\left(G^\ell\right)\le k + R_k(\lceil\frac{n-k}{2}\rceil)$, as desired.
\end{proof}

The recurrence~$R_k(n)$ is studied in~\cite{Gruber10}, where 
also the inequality $R_k(n) \le k \cdot \log(n/k)$ is derived.\footnote{Here $\log$ denotes the binary logarithm.}
We thus have the following bound:

\begin{corollary}\label{cor:cr-separator}
Let~$G$ be a loop-free digraph with~$n$ vertices and weak separator number 
at most~$k$. Then $\crank(G)\le k \cdot \log(n/k)-1$.\qed
\end{corollary}

This inequality is sharp already in the undirected case, 
see~\cite{Gruber10}. Previously, a looser bound comparing cycle rank 
to a similar notion of weak separator number was 
given in~\cite{GH08a}. It is easy to see that 
Corollary~\ref{cor:cr-separator} improves upon the previous bound.

\medskip

We turn to the comparison with directed pathwidth.
That measure was introduced by Reed, Seymour and Thomas 
(cf.~\cite{Barat06}) as a generalization 
of pathwidth to digraphs.
\begin{definition}
For a digraph $G=(V,E)$, a {\em directed path decomposition} of~$G$ is a 
sequence $W_1W_2\cdots W_r$ of subsets of~$V$, called {\em bags}, 
such that
\begin{enumerate}[a)]
\item each vertex is contained in at least one bag, 
\item for all $i<j<k$ holds~$W_i \cap W_k \subseteq W_j$,
and 
\item for each edge $(u,v)$ in~$E$, there is a bag containing both endpoints, 
or there exist $i,j$ with $i<j$ 
such that the tail~$u$ is in~$W_i$ and the head~$v$ is
in~$W_j$. 
\end{enumerate}
The {\em width} of a directed path decomposition is defined as the maximum
cardinality among all bags minus~$1$. The directed pathwidth is defined as
the minimum width among all directed path decompositions for~$G$. 
\end{definition}
A directed path decomposition is {\em normal}, if adjacent bags may differ in 
at most one vertex, and it is easy to transform a directed path decomposition 
into a normal one. Based on normal path decompositions, it is not difficult to 
derive the following result:
\begin{lemma}
Let~$G$ be a digraph. Then $\snum(G) \le \dpw(G)$. \qed 
\end{lemma}

How does cycle rank relate to directed pathwidth? 
We can answer this using directed elimination forests.
\begin{lemma}\label{thm:dpw-crank}
Let~$G$ be a digraph. Then $\dpw(G) \le \crank(G)$. 
\end{lemma}
\begin{proof}
We prove by induction that each directed elimination 
forest of height~$k$ for~$G$ can be transformed into a 
directed path decomposition for~$G$ of width at most~$k$.

If~$k=0$, then~$G$ is acyclic, and thus clearly admits 
a directed path decomposition of with~$0$.

For the induction step, assume the directed elimination 
forest for $G$ has roots $(x_1,C_1)$, $(x_2,C_2)$,\ldots, 
$(x_r,C_r)$, with the strongly connected components~$C_i$
in topological order. Let $G_i = G[C_i]-x_i$. 
Then $G_i$ has cycle rank at most $k-1$. By induction assumption, 
each digraph $G_i$ admits a directed path decomposition of 
width at most $k-1$. By adding the vertex~$x_i$ to each bag in the 
respective decomposition for $G_i$, we obtain a directed 
path-decomposition for $G[C_i]$. Concatenating the~$r$ 
individual directed path decompositions 
while respecting the above topological order, 
we obtain a directed path decomposition of width at 
most~$k$ for~$G$, as desired. 
\end{proof}

Altogether, we have derived the following chain of inequalities:
\begin{theorem}\label{thm:bounds}
Let~$G$ be a loop-free digraph with~$n$ vertices and weak separator 
number~$k$. Then 
\[ k \le \dpw(G) \le \crank(G)\le k \cdot \log(n/k)-1. \qquad\qed\] 
\end{theorem}

Quite a few more structural complexity measures on digraphs were 
studied recently, such as directed tree-width, DAG-width, and Kelly-width. 
As detailed in~\cite{KKK11}, each of these measures 
is bounded below by a function that is linear in the weak separator 
number\footnote{The notion used in~\cite{KKK11} corresponds to 
our notion of weak separator number up to a constant factor.}. 
On the other hand, all of those are bounded above by the directed 
pathwidth (cf.~\cite{KKK11}), so Theorem~\ref{thm:bounds} will 
also serve for comparing them with cycle rank, and with 
weak separator number. 

\section{Computational Aspects of Cycle Rank}\label{sec:cr-algo}

\subsection{Computational Complexity}

We turn to algorithmic questions. First,
we classify the computational complexity of the decision 
problem {\bf CYCLE RANK}: Given a digraph~$G$ and an 
integer~$k$, determining whether the cycle rank of~$G$ is 
at most~$k$. 

\begin{theorem}\label{thm:cr-npc} 
The {\bf CYCLE RANK} problem is $\NP$-complete, and this still holds
when requiring that the input digraph is strongly connected.
\end{theorem}
\begin{proof}
Membership in $\NP$ can be seen by the equivalent definition
using directed elimination forests: Let $G=(V,E)$ denote 
the given digraph. Every elimination forest for~$G$ contains at 
most~$|V|$ tree vertices, and each tree vertex is of size is at most~$|V|$.  
A nondeterministic polynomial-time bounded Turing machine can guess 
such a witness, and then verify 
that it indeed constitutes an elimination forest of height at most~$k$. 

For $\NP$-hardness, we use a corresponding result known for the 
undirected case. Given a symmetric loop-free digraph~$G$, it 
is easy to see (e.g. by~\cite[Lem.~2.2]{NO06})
that an undirected elimination forest of height~$k+1$ in the sense
of~\cite{BGHK95,NO06} corresponds to a directed elimination forest of 
height $k$ in our sense (the term $+1$ accounts for 
the slightly different definition of {\em height} used in~\cite{NO06}). 
However, determining the minimum height among all
undirected elimination forests is $\NP$-complete, 
also for (strongly) connected undirected graphs~\cite{BGHK95}.
\end{proof}
Using tools from formal language theory,
we will prove later that $\NP$-hardness still holds for digraphs of 
maximum outdegree at most~$2$ and of maximum total degree at most~$4$. 

\subsection{Approximate Computation}

How to cope with this negative result? One possibility
is to look for an approximate solution. Indeed, it is known 
that for undirected graphs, the cycle rank problem 
admits an input-dependent 
polynomial-time approximation algorithm~\cite{BGHK95}. 
In the following, we devise a more general 
approximation algorithm, which covers also the 
case of unsymmetric digraphs. The basic pattern of 
our algorithm for directed cycle rank is again 
divide-and-conquer along separators.  

\begin{theorem}\label{thm:approx-cr}
The {\bf CYCLE RANK} problem admits a polynomial-time 
approximation within a factor of~$O((\log n)^{3/2})$.
\end{theorem}
\begin{proof}
The following recursive procedure computes a directed 
elimination forest for the induced subgraph~$G[W]$,
where~$W \subset V$ is passed as parameter to the 
procedure. 

If~$G[W]$ consists of several strongly connected components, 
apply the procedure recursively to each of these; 
The union of these results gives a directed elimination 
forest for~$G[W]$. 

Otherwise, use the polynomial-time algorithm 
from~\cite[Corollary~2.25]{KKK11} 
to find a small vertex subset~$S \subseteq W$ in~$G[W]$ with the 
property that every strongly connected component of~$G[W]-S$ has at 
most~$\frac34|W|$ vertices.
Then pass the digraph~$G[W]-S$ as parameter to the recursive procedure.
Upon returning, the directed elimination forest~$F$ returned for~$G[W]-S$
is then extended, one by one, for each vertex~$s$ from~$S$. 

More precisely, put the elements of~$S$ in arbitrary order. 
Then for given~$s$ in~$S$, let~$X$ denote the set of vertices 
occurring before~$s$. Assuming we have already computed a directed
elimination forest for~$G[W \cup X]$, we now show how to extend 
this to a forest for~$G[W \cup X \cup s]$. Initially, the set~$X$ 
is empty, and we proceed for each~$s$ until $X=S$. 
Let $C_1,\ldots C_p$ denote those strongly connected components 
of the digraph~$G[W \cup X]$ 
for which $G[W \cup \{s\} \cup \bigcup_i C_i]$ is strongly connected,
and let $D_1,\ldots,D_r$ denote the remaining strongly 
connected components in~$G[W \cup X]$.
The elimination forest for~$G[W \cup X]$
contains an elimination tree for each~$G[C_i]$, and for each~$G[D_i]$.
Make up a new root~$(s, X \cup \{s\})$, and attach the directed elimination 
trees for the digraphs~$G[C_i]$ as children to that new root. This 
gives a directed elimination tree for~$G[W \cup \{s\}\cup \bigcup_i C_i]$. 
The union of this tree with the directed elimination trees 
for the strongly connected components $D_1, \ldots, D_r$ 
yields a directed elimination forest for~$G[W \cup X \cup s]$.
This completes the description of the subroutine for 
extending the forest.

The recursion terminates 
as soon as the size of~$W$ decreases below~$\beta (\log n)^{3/2}$.
In this case,
simply return an (arbitrary) directed elimination 
forest for~$G[W]$. 

Here, the number~$\beta$ is a fixed, suitably
chosen, constant coming from the analysis below. This completes 
the description of the algorithm.

It remains to analyze the above algorithm. It is readily checked 
that the algorithm returns an elimination forest for~$G$. 
For the performance guarantee, those recursive calls that simply 
partition the graph into strongly connected components do not add to the
height of the resulting forest; if we restrict our attention 
to these recursive calls that compute a suitable vertex 
subset~$S$, the depth of the recursion tree is~$O(\log n)$. 
At each such step,
we can find in polynomial time a suitable set~$S$ of size at most 
$\beta k \sqrt{\log n}$, where~$k$ is the 
directed pathwidth of~$G$, and~$\beta$ is some 
known constant (cf.~\cite[Corollary~2.25]{KKK11}).
The recursion terminates with an elimination forest of 
height at most~$\beta \cdot (\log n)^{3/2}$. Thus the overall 
height is bounded by 
\[\beta\cdot k \cdot\sqrt{\log n} \cdot O(\log n) + \beta \cdot (\log n)^{3/2} 
= O(k\cdot(\log n)^{3/2}),\] 
where~$k$ is the directed 
pathwidth of~$G$. By Lemma~\ref{thm:dpw-crank}, 
we have~$k \le \crank(G)$. In this way, we have 
a polynomial-time~$O((\log n)^{3/2})$-approximation 
for cycle rank.
\end{proof}

The above performance guarantee matches the best 
previous result known for the undirected 
case~\cite{ACMM05}. For other digraph 
complexity measures, such as D-width and 
directed pathwidth, approximation algorithms in a 
similar vein were recently given in~\cite{KKK11}.

\subsection{Exact Computation}

In certain circumstances, an 
approximation guarantee within a factor~$O((\log n)^{3/2})$
may not suffice. Thus we also take a look at exact algorithms 
for computing the cycle rank.

The na\"{\i}ve algorithm for determining cycle rank according to 
Definition~\ref{defn:cr} requires inspecting~$n!$ 
possibilities on a graph with~$n$ vertices, 
as witnessed by the complete graph~$K_n$. While one may not expect a
polynomial-time algorithm, we can still do much better:
\begin{theorem}
The cycle rank of an $n$-vertex digraph can be computed in 
time and space~$O^*(2^n)$. 
\end{theorem}
\begin{proof}
We show how the characterization of 
the cycle rank of a digraph $G=(V,E)$ in terms of
the directed elimination forests from Definition~\ref{defn:elimination-forest} 
can be turned into a dynamic 
programming scheme. We only consider the case~$G$ itself is nontrivially strongly connected---otherwise, we obtain the cycle 
rank by taking the minimum among the cycle ranks of the nontrivial strongly connected components of~$G$. 
For a nontrivial strongly connected subset of 
vertices~$X \subseteq V$ and a vertex~$x\in X$, let $\crank(x,X)$ denote 
the minimum height among all elimination forests for $G$ with root $(x,X)$.
Then $\crank(G) = \min_{v\in V} \crank(v,V)$, so it suffices to 
design an algorithm computing $\crank(v,V)$ for each $v\in V$.
By inspecting Definition~\ref{defn:elimination-forest}, we obtain 
the recurrence 
\begin{equation}\label{eqn:recurrence}
\crank(x,X) = \begin{cases} 1 & \mbox{ if } G[X]-x \mbox{ is acyclic }\\
1+\max_{Y} \min_{y \in Y} \crank(y,Y) & \mbox{ otherwise }\end{cases}
\end{equation}
Here~$Y$ runs over all nontrivial strongly connected components of $G[X]-x$ (of which there
can be at most $|X|-1$). Using the classic trick 
of memoization (see~\cite{KT05}),
this recurrence can be easily transformed into a dynamic programming scheme 
with memoization that runs in time $|\mathfrak{S}|\cdot n^{O(1)}$, 
where~$\mathfrak{S} \subseteq 2^{V}$ is the set of strongly connected subsets
of the digraph~$G$. 
\end{proof}

The reader is invited to try out the above algorithm 
for the digraph depicted in Figure~\ref{fig:digraph}. 
The bottleneck in the above algorithm is the requirement 
of computing and storing the cycle rank for all elements of $\mathfrak{S}$,
namely of the family of strongly connected subsets in the input digraph. 
For a complete digraph, we have $|\mathfrak{S}|=2^n$, but 
this bound can no longer be reached for digraphs of bounded 
maximum outdegree. For undirected graphs of maximum degree~$d$, 
a nontrivial bound on the number of (weakly) connected 
subsets was established recently in~\cite{BHKK08}. As it turns out,
their bound allows the following generalization 
to the theory of digraphs, in that the original proof 
carries over with obvious modifications:

\begin{lemma}\label{lem:count-strong-subsets}
Let~$G$ be a digraph of order $n$ with maximum outdegree at most~$d$.
Then the number of strongly connected subsets of~$V$ is 
at most $\gamma^n + n$, with $\gamma = (2^{d+1}-1)^{1/(d+1)}$.
In particular, for $d=2$, we have $\gamma \doteq 1.9129$.\qed
\end{lemma}

On digraphs of bounded outdegree, we thus obtain the following improved 
bound on the running time of the above algorithm:

\begin{theorem}\label{thm:algos-crank}
Let~$G$ be a digraph of order~$n$ with constant 
maximum outdegree~$d$.
Then the cycle rank of~$G$ can be computed in time and
space $O^*\left((2-\epsilon)^n\right)$, where $\epsilon$ 
is a constant depending on~$d$. In particular, for 
digraphs of maximum outdegree~$2$, 
the cycle rank can be computed in time and space~$O^*(1.9129^n)$.\qed
\end{theorem}  

It seems that Lemma~\ref{lem:count-strong-subsets} has a 
host of algorithmic consequences. For illustration, recall 
that a vertex subset $S \subseteq V$ of a digraph $G$ is a 
{\em directed feedback vertex set}, if removing~$S$ 
from~$G$ leaves an acyclic digraph. Off the cuff, 
we can devise an exact algorithm for minimum 
directed feedback vertex set on sparse digraphs.
\begin{theorem}\label{thm:algo-dfvs}
Let~$G$ be a digraph of order~$n$ with constant 
maximum outdegree~$d$.
Then a minimum directed feedback vertex set of~$G$ can be computed in time and
space $O^*\left((2-\epsilon)^n\right)$, where $\epsilon$ 
is a constant depending on~$d$. In particular, for 
digraphs of maximum outdegree~$2$, 
a minimum directed feedback vertex set 
can be computed in time and space~$O^*(1.9129^n)$.\qed
\end{theorem}
\begin{proof}
By duality, the task of enumerating all minimal directed 
feedback vertex sets is equivalent to enumerating all maximal
acyclic subsets, that is, maximal vertex subsets that induce 
a directed acyclic graph. Here, ``minimal'' and ``maximal'' are 
meant with respect to set inclusion. 

Since there is an algorithm enumerating all 
minimal directed feedback vertex sets 
(or, equivalently, all maximal acyclic subsets) 
with polynomial delay~\cite{SchwikowskiSpeckenmeyer02}, it only 
remains to derive a 
combinatorial bound on the number of such 
sets. 
A strongly connected subset~$S \subset V$ in $G$ is called a 
{\em minimal strongly connected subset}, if~$S$ contains 
a vertex~$v$ such that $S-v$ is an acyclic subset. 
Clearly, in this case, $S-v$ is a maximal acyclic 
subset. Thus, each minimal strongly connected 
subset~$S$ will give rise to at most $|S| \le n$ 
maximal acyclic subsets; and each maximal acyclic subset can
be obtained in this way from a minimal strongly connected subset.
Thus the total number of maximal acyclic subsets in~$G$ 
is at most $n$ times the number of (minimal) strongly connected 
subsets in~$G$. The result now follows with 
Lemma~\ref{lem:count-strong-subsets}.
\end{proof}
The above running time looks reasonable if we consider 
the following facts: First, even on digraphs of maximum outdegree at 
most~$2$, the problem is $\NP$-complete~\cite[Problem~GT7]{GJ79}.
Second, the fastest known exact algorithm for digraphs of unbounded
outdegree~\cite{Razgon07} runs in time~$O^*(1.9977^n)$. Third, 
easy examples show that digraphs with outdegree~$2$ can have at 
least~$1.4142^{n}$ minimal directed feedback vertex sets~\cite{SchwikowskiSpeckenmeyer02}.

\section{Star Height of Regular Expressions}\label{sec:star-height}
As it turns out, the cycle rank of digraphs is intimately 
related to structural and descriptional complexity aspects of 
regular expressions. 
The star height of a regular language~$L$, denoted by $\starheight(L)$, 
is defined as the minimum nesting depth of stars in 
any regular expression describing~$L$. The following relation between star
height and the cycle rank of nondeterministic finite automata (NFAs)
was shown already in the seminal paper on star height.
\begin{theorem}[Eggan's Theorem]
Let~$L$ be a regular language. Then 
$$\starheight(L)= \min \{\,\crank(A) \mid A \mbox{ an NFA accepting }L\,\}$$
\end{theorem}
Here, $\crank(A)$ denotes the cycle rank of the digraph underlying 
the transition structure of~$A$. 

\medskip

As an aside, Eggan's Theorem was recently used 
to obtain a powerful lower bound technique for the minimum required 
length of regular expressions for a given regular language:
\begin{lemma}[Star Height Lemma,~\cite{GH08a}]
Let~$L$ be a regular language. If $L$ admits a regular expression 
of length~$n$, then $n \ge 2^{\Omega\left(\starheight(L)\right)}$.  
\end{lemma}
The gist of the proof is that each regular expression can be converted 
into an equivalent NFA of comparable size, 
but whose transition structure is only poorly connected. 
The result then follows using Eggan's Theorem. In~\cite{GH08a}, this
method was used to prove the unexpected result
that complementing regular languages can cause 
a doubly-exponential blow-up in the minimum required regular expression length.

\medskip

Of course, the minimum in Eggan's Theorem is taken over infinitely many NFAs, 
and indeed for more than two decades, it 
was unknown whether there exists an algorithm deciding 
the {\bf STAR HEIGHT} problem: 
given a deterministic finite automaton (DFA)~$A$ and an integer~$k$, 
determine whether the star height of $L(A)$ is at most~$k$, a question raised
in~\cite{Eggan63}. Although the problem is now known to be decidable,
the best known upper bound\footnote{The noted upper bound holds more 
generally for a given NFA if also an NFA accepting 
the complement language is provided as part of the input.
Recall that complementing a DFA does not affect its size, whereas 
complementing an NFA can cause an exponential blow-up 
in the required number of states~\cite{HK03}.
} 
to date is $\EXPSPACE$~\cite{Kirsten10}. 
To the best of our knowledge, nontrivial lower bounds 
are known {\em only} for the case where the input is specified 
succinctly, as an NFA: Determining the star height of a language 
specified as an NFA is $\PSPACE$-hard~\cite{HR78}.
Yet, as illustrated in~\cite{HR78}, a large multitude of natural questions 
about the language accepted by a given NFA is $\PSPACE$-hard, whereas
the corresponding question often become computationally easy 
if a DFA is given. Therefore,
such a hardness result renders more service to understanding the effect of 
succinct input descriptions than 
to understanding the computational nature of the core problem at hand. 
That is why we deliberately stick to the convention to specify the 
input as a DFA.

Here we settle the complexity of the star height problem 
for a subclass of the regular languages, namely the bideterministic languages.
The decision problem {\bf BIDETERMINISTIC STAR HEIGHT} is defined as follows:
Given a bideterministic finite automaton~$A$ and an integer~$k$, 
decide whether the star height of~$L(A)$ is at most~$k$.

Bideterministic finite automata have the special property that
the star height problem of bideterministic languages
boils down to determining the cycle rank of a digraph. 
The following theorem is proved in~\cite{McNaughton67}:
\begin{theorem}[McNaughton's Theorem]\label{thm:McNaughton}
Let~$L$ be a bideterministic language, and let~$A$ be the minimal trim
(i.e., without a dead state) DFA accepting~$L$.
Then $\starheight(L) = \crank(A)$. 
\end{theorem}
On the positive side, the algorithmic results from the previous section 
easily translate to a formal language setup 
using McNaughton's Theorem. For approximating {\bf STAR HEIGHT}, 
we have to resort to Eggan's Theorem, 
giving only an $O(n)$-approximation. In the 
bideterministic case, we have the following 
counterpart to Theorem~\ref{thm:approx-cr}:
\begin{theorem}\label{thm:approx-sh}
The {\bf BIDETERMINISTIC STAR HEIGHT} problem admits a polynomial-time 
approximation within a factor of~$O((\log n)^{3/2})$.\qed
\end{theorem}
We also have a natural counterpart to Theorem~\ref{thm:algos-crank}:
\begin{theorem}\label{thm:algos-sh}
Let~$A$ be a bideterministic finite automaton with~$n$ states 
over an input alphabet of size~$k$.
Then the star height of~$L(A)$ can be computed exactly, in time and
space $O^*\left((2-\epsilon)^n\right)$, where $\epsilon$ 
is a constant depending on~$k$. In particular, for the case of 
binary input alphabets, the star height can be 
computed in time and space~$O^*(1.9129^n)$.\qed
\end{theorem}
On the negative size, also the 
$\NP$-hardness result for {\bf CYCLE RANK} translates to 
its language-theoretic counterpart.
Moreover, we show that already 
the case of binary input alphabets is that hard:
\begin{theorem}\label{thm:sh-npc}
The {\bf BIDETERMINISTIC STAR HEIGHT} problem is $\NP$-complete, 
and this still holds when restricted to bideterministic 
automata over binary input alphabets.
\end{theorem}
\begin{proof}
We first show $\NP$-completeness for the case of unbounded alphabet size,
and then provide a polynomial-time reduction to the case of binary alphabets.  

For membership in $\NP$, we use McNaughton's Theorem 
(Theorem~\ref{thm:McNaughton}) to reduce the problem 
to {\bf CYCLE RANK}, and the latter is in~$\NP$ 
by Theorem~\ref{thm:cr-npc}.

To establish $\NP$-hardness, we reduce from the problem of determining 
for a strongly connected digraph $G=(E,V)$ and an integer~$k$ whether the
cycle rank is at most~$k$, which is $\NP$-hard by Theorem~\ref{thm:cr-npc}. 
For a vertex~$v$ in~$V$,
define 
\[L(G,v) =\{\,w\in E^* \mid w \mbox{ is a walk in $G$ starting and ending in } v \,\}.\] 
A deterministic finite automaton~$A$ accepting $L(G,v)$
has~$V$ as set of states and for each edge $(x,y) \in E$ a transition labeled 
$(x,y)$ from~$x$ to~$y$. The start and only accepting state is~$v$. It is
readily verified that~$A$ accepts $L(G,v)$, is bideterministic, and that 
$A$ is the minimal trim DFA for this language. 
By construction, 
$\crank(A)=\crank(G)$, and $\crank(A)=\starheight(L)$ by Theorem~\ref{thm:McNaughton}. This completes the $\NP$-completeness proof for unbounded alphabet size.

We turn to the case of binary alphabets. Given an instance~$(A,k)$ 
of {\bf BIDETERMINISTIC STAR HEIGHT}, we construct in polynomial time 
a bideterministic finite automaton~$B$ over the alphabet~$\{a,b\}$, 
such that the star height of~$B$ equals the star height of~$A$. 
Assume the input alphabet of~$A$ is~$\Sigma=\{a_1,a_2,\ldots,a_r\}$. 
The automaton~$B$ will accept the homomorphic image of~$L(A)$ under 
the homomorphism $\rho:\Sigma\to\{a,b\}$ given by $\rho(a_i)=a^ib^{r+1-i}$,
for $1\le i\le r$. It is known~\cite{McNaughton69} that~$\rho$ 
preserves star height, that is, for every 
regular language~$L$, the image of~$L$ under~$\rho$ is of the same 
star height as~$L$. It remains to construct a bideterministic 
automaton~$B$ accepting~$\rho(L(A))$ in polynomial time: automaton~$B$
will have the states of~$A$, plus some extra states. 
For each state~$q$ copied from~$A$, we add~$r$
states $q^{+}_{1},q^{+}_{2},\ldots q^{+}_{r}$ and~$r$ more 
states $q^{-}_{1},q^{-}_{2},\ldots q^{-}_{r}$ to the state set of~$B$. 
The transition relation
of~$B$ is given by requiring that whenever there 
is a transition~$p\trans{a_i}q$ 
in~$A$, then~$B$ admits the sequence of transitions
\[p \trans{a} p^+_1 \trans{a}p^+_2\cdots \trans{a} p^+_{i} \trans{b} 
q^-_{r-i}\trans{b}\cdots q^-_{2}\trans{b}q^-_{1}\trans{b} q.\]
There are no other transitions in~$B$. By construction,~$B$
accepts~$\rho(L(A))$. It is easily verified 
that if~$A$ is bideterministic, then so is~$B$.
\end{proof} 

Returning again to {\bf CYCLE RANK}, we observe that
the digraph underlying a bideterministic automaton over
a binary alphabet always 
has maximum outdegree at most~$2$ 
and maximum total degree at most~$4$.
The correspondence
given by McNaughton's Theorem
between bideterministic automata and digraphs yields the 
following consequence:
\begin{corollary}\label{cor:cr-np-outdeg-2}
The {\bf CYCLE RANK} problem restricted 
to digraphs of maximum outdegree at most~$2$ and total 
degree at most~$4$ remains $\NP$-complete.\qed
\end{corollary}

\section{Conclusion}\label{sec:conclusion}

In this work, we explored measures for the complexity of digraphs,
and their applications. We paid particular attention to
the cycle rank of digraphs 
and its relation to other digraph complexity measures, 
as well as its connection to the star height of regular languages. 
A tabular summary of our main algorithmic results is given 
in the Appendix.

Regarding cycle rank, the undirected case seems to be much
better understood than the general case. 
An intriguing open question
is whether the cycle rank problem is fixed-parameter tractable. This is known
to be the case on undirected graphs, see~\cite{BDJKKMT98}.
 
Regarding the star height problem, the picture is even less 
clear. The main problem,
namely the decidability status, has been settled for more than~$20$ 
years now. Still, the computational complexity
of this problem is not well understood. From the viewpoint of 
a computational complexity, 
we studied the ``easiest hard case'', and showed that 
(the non-succinct version of) this problem is $\NP$-hard. 
Currently the best upper bound~\cite{Kirsten10} is $\EXPSPACE$. 
Tightening the eminent gap between these bounds 
is surely a challenging theme for further research.

\subsubsection*{Acknowledgment}
The author would like to thank Markus Holzer for carefully 
reading an earlier draft of this paper, 
and for providing some valuable suggestions.

\bibliographystyle{plain}
\bibliography{star-height}
\appendix
\section{Appendix}

\begin{center}{
\renewcommand{\arraystretch}{1.3}
\renewcommand{\tabcolsep}{2mm}

\centering
\begin{tabular}{|r@{  }p{0.7\textwidth}|}
\hline\hline
\multicolumn{2}{|c|}{{\bf CYCLE RANK}}\\ 
\hline
Instance. & A digraph $G$ and an integer $k$. \\
Question. & Is the cycle rank of $G$ at most $k$? \\
Good news. & Approximable within~$O((\log n)^{3/2})$ in polynomial time 
(Thm.~\ref{thm:approx-cr}). Exact solution can be computed in 
time~$O^*(1.9129^n)$ for digraphs with maximum outdegree at most~$2$; 
and for unbounded outdegree in time $O^*(2^n)$ 
(Thm.~\ref{thm:algos-crank}).\\
Bad news. & $\NP$-complete (Thm.~\ref{thm:cr-npc}). Problem 
is $\NP$-hard already
for digraphs of maximum outdegree~$2$ and maximum 
total degree~$4$ (Cor.~\ref{cor:cr-np-outdeg-2}); 
$\NP$-hard also for some classes of undirected graphs 
(e.g., bipartite and cobipartite)~\cite{BDJKKMT98}. \\
\hline\hline
\end{tabular}

\medskip

\centering
\begin{tabular}{|r@{  }p{0.7\textwidth}|}
\hline\hline
\multicolumn{2}{|c|}{{\bf DIRECTED FEEDBACK VERTEX SET}}\\ 
\hline
Instance. & A digraph $G$ and an integer $k$. \\
Question. & Does $G$ admit a directed feedback vertex set 
of cardinality at most $k$? \\
Good news. & For digraphs with maximum outdegree at most~$2$,
exact solution can be computed in time~$O^*(1.9129^n)$   
(Thm.~\ref{thm:algo-dfvs}); 
and in time $O^*(1.9977^n)$ for unbounded outdegree~\cite{Razgon07}.
Problem is fixed-parameter tractable~\cite{CLLOR08}.\\
Bad news. & $\NP$-complete, already
for digraphs of maximum outdegree~$2$~\cite[Problem~GT7]{GJ79}. \\
\hline\hline
\end{tabular}

\medskip

\centering
\begin{tabular}{|r@{  }p{0.7\textwidth}|}
\hline\hline
\multicolumn{2}{|c|}{{\bf BIDETERMINISTIC STAR HEIGHT}}\\ 
\hline
Instance. & A bideterministic finite automaton $A$ and an integer $k$. \\
Question. & Is the star height of $L(A)$ at most $k$? \\
Good news. & Approximable within~$O((\log n)^{3/2})$ in polynomial time 
(Thm.~\ref{thm:approx-sh}). Exact solution can be computed in 
time~$O^*(1.9129^n)$ for binary alphabets; 
and for unbounded alphabet size in time $O^*(2^n)$ 
(Thm.~\ref{thm:algos-sh}).\\
Bad news. & $\NP$-complete; $\NP$-hardness holds already for 
binary alphabets (Thm.~\ref{thm:sh-npc}). \\
\hline\hline
\end{tabular}

\medskip

\centering
\begin{tabular}{|r@{  }p{0.7\textwidth}|}
\hline\hline
\multicolumn{2}{|c|}{{\bf STAR HEIGHT}}\\ 
\hline
Instance. & A deterministic finite automaton $A$ and an integer $k$. \\
Question. & Is the star height of $L(A)$ at most $k$? \\
Good news. & Problem is decidable~\cite{Hashiguchi88}. 
Exact solution can be computed within 
exponential space and doubly exponential time~\cite{Kirsten10}.\\
Bad news. & $\NP$-hard, already for binary alphabets (Thm.~\ref{thm:sh-npc}). 
Problem is $\PSPACE$-hard if input given by an nondeterministic 
finite automaton in place of a deterministic one~\cite{HR78}.\\
\hline\hline
\end{tabular}
}
\end{center}
\end{document}